\documentclass[envcountsame]{llncs}
\usepackage[colorlinks]{hyperref}

\usepackage{version}
\usepackage{wrapfig}
\includeversion{report}\excludeversion{conference}

\usepackage{fancybox}
\usepackage{amssymb}
\usepackage{amsfonts}
\usepackage{amsmath}
\usepackage{scalefnt}
\usepackage{url}
\usepackage{verbatim}
\usepackage{color}
\usepackage{subfig}
\usepackage{stmaryrd}
\usepackage{tikz-qtree}%
\usepackage{colortbl}

\usepackage{xspace}
\usepackage[latin1]{inputenc}

\usepackage[ruled,vlined,linesnumbered]{algorithm2e}

\usetikzlibrary{patterns}

\newcommand{\tuple}[1]{\langle#1\rangle}

\newcommand{\setof}[2]{\left\{#1\,\middle|\:#2\right\}}
 \newcommand{\gpid}{\textsc{GPiD}\xspace}

 \newcommand{\Ilinva}{\textsc{Ilinva}\xspace}
 \newcommand{\Propagate}{\textsc{Ind}\xspace}
 
 \newcommand{\true}{\top}
 \newcommand{\interv}[2]{\llbracket#1\mathrel{{.}\,{.}}\nobreak#2\rrbracket}
 
 \newcommand{\modelst}{\models_{\theory}}
 \newcommand{\theory}{{\cal T}}
 \newcommand{\whyml}{\textsc{WhyML}\xspace}
 
\newcommand{\zthree}{\textsc{Z3}\xspace}
\newcommand{\smtlib}{\textsc{SMTlib2}\xspace}
\newcommand{\abdulot}{\textsc{Abdulot}\xspace}
\newcommand{\ilinva}{\texttt{Ilinva}\xspace}
\newcommand{\cvcf}{\textsc{CVC4}\xspace}
\newcommand{\whyt}{\texttt{Why3}\xspace}
 \newcommand{\why}{\whyt}
 
\newcommand{\github}{\textsc{GitHub}\xspace}
\newcommand{\ccode}{\textsc{C}\xspace}
\newcommand{\java}{\textsc{Java}\xspace}

\newcommand{\minisat}{\textsc{MiniSAT}\xspace}
\newcommand{\altergo}{\textsc{AltErgo}\xspace}
\newcommand{\cpp}{\textsc{C++}\xspace}
\newcommand{\ie}{\emph{i.e.}\xspace}
\newcommand{\eg}{\emph{e.g.}\xspace}

\newcommand{\rtnosol}[1]{\cellcolor{lightgray} $#1$}
\newcommand{\rtasol}[1]{$#1$}
\newcommand{\rtto}[1]{\cellcolor{gray} $#1$}
\newcommand{\rtnoth}[1]{\cellcolor{black} $#1$}
\newcommand{\compact}[1]{{\small #1}}

\renewcommand*{\thefootnote}{\fnsymbol{footnote}}

\newcommand{\todo}[1]{\textbf{\color{red}TODO: #1}}
\newcommand{\crrm}[1]{}

\begin{document}

\title{\Ilinva: Using Abduction to Generate Loop Invariants}
\author{Mnacho Echenim \and Nicolas Peltier \and Yanis Sellami}
\institute{Univ. Grenoble Alpes, CNRS, LIG, F-38000 Grenoble France \email{[Mnacho.Echenim|Nicolas.Peltier|Yanis.Sellami]@univ-grenoble-alpes.fr}}

\maketitle

\begin{abstract}
We describe a system to prove properties of programs. 
The key feature of this approach is a method to automatically synthesize inductive invariants of the loops contained in the program. The method is generic, \ie, it applies to a large set of programming languages and application domains; and lazy, in the sense that 
it only generates invariants that allow one to derive the required properties.
It relies on an existing system called \gpid for abductive reasoning modulo theories  \cite{EPS18}, and on the platform for program verification \why \cite{filliatre13esop}.
Experiments show evidence of the practical relevance of our approach.
\end{abstract}

\renewcommand*{\thefootnote}{\arabic{footnote}}

\newcommand{\aprog}{{\tt P}}
\newcommand{\aloop}{{\tt L}}

\newcommand{\varlist}{\texttt{L}}

\newcommand{\aformA}{\phi}
\newcommand{\aformB}{\psi}

\newcommand{\acond}{{\tt C}}
\newcommand{\Hoare}[3]{#1\ \{ #2 \}\ #3}

\newcommand{\niko}[1]{{\color{red} Niko: #1}}
\newcommand{\mnacho}[1]{{\color{blue} Mnacho: #1}}
\newcommand{\yanis}[1]{{\color{green} Yanis: #1}}

\newcommand{\removecomments}{
	\renewcommand{\niko}[1]{}
	\renewcommand{\mnacho}[1]{}
	\renewcommand{\yanis}[1]{}}


\section{Introduction}

Hoare logic -- together with strongest post-conditions or weakest pre-conditions calculi -- allow one to verify properties of programs 
defined by bounded sequences of instructions \cite{Hoare:1969:ABC:363235.363259}. Given a pre-condition $\aformA$ satisfied by the inputs of program $\aprog$, 
algorithms exist to compute the strongest formula $\aformB$ such that $\Hoare{\aformA}{\aprog}{\aformB}$ holds, meaning that if $\aformA$ holds initially then $\aformB$ is satisfied after $\aprog$ is executed, and any formula $\aformB'$ that holds after $\aprog$ is executed is such that $\aformB\models \aformB'$. To check that the final state satisfies some formula $\aformB'$, we thus only have to check that 
$\aformB'$ is a logical consequence of $\aformB$.
However, in order to handle programs containing loops, it is necessary to associate each loop occurring within the program with an {\em inductive invariant}. An inductive invariant for a given loop $\aloop$ is a formula that  holds every time the 
program enters $\aloop$ (\ie, it must be a logical consequence of the preconditions of $\aloop$), and is preserved by the sequence of instructions in $\aloop$.
Testing whether a formula is an inductive invariant is a straightforward task, and the difficulty resides in  generating candidate invariants.
These can be supplied by the programmer, but this is a rather tedious and time-consuming task; for usability and scalability,  
it is preferable to generate those formulas automatically when possible.
In this paper, we describe a system to generate such invariants in a completely automated way, via abductive reasoning modulo theories, based on the methodology developed in \cite{DBLP:conf/oopsla/DilligDLM13}. Roughly speaking, the algorithm works as follows. Given a program $\aprog$ decorated with a set of assertions that are to be established, all loops are first assigned the same candidate invariant $\true$. These invariants
are obviously sound:  they hold before the loops and are preserved by the sequence of instructions in the loop; however they are usually not strong enough to prove
the assertions decorating the program. They are therefore strengthened by adding hypotheses that are sufficient to ensure that the assertions hold; these hypotheses are generated by a tool that performs \emph{abductive inferences}, and the strengthened formulas are \emph{candidate invariants}. Additional strengthening
steps are taken to guarantee that these candidates are actual invariants, \ie, that they
are preserved by the sequence of instructions in the loop. These steps are iterated until a set of candidate invariants that are indeed inductive is obtained. 

We rely on two existing systems to accomplish this task. The first one is \why (see, \eg, \url{http://why3.lri.fr/} or \cite{filliatre13esop}), a well-known and widely-used  platform for deductive program verification that is used to compute verification conditions and verify assertions. The second system, \gpid,
  is
designed to generate implicants\footnote{An implicant of a formula $\psi$ is a formula $\phi$ such that $\phi\models \psi$. It is the dual notion of that of implicates}
of quantifier-free formulas modulo theories \cite{EPS18}. This system is used as an abductive reasoning procedure, thanks to the following property: if $\phi\not\models\psi$, finding a hypothesis $\phi'$ such that $\phi\wedge \phi'\models \psi$ is equivalent to finding $\phi'$ such that $\phi'\models \phi\Rightarrow \psi$. 
{\gpid} is  generic, since it only relies on 
the existence of a decision procedure for the considered theory (counter-examples
are exploited when available to speed-up the generation of the implicants when available). 
Both systems are connected in the \Ilinva framework.

\paragraph*{Related Work.}

A large number of different techniques have been proposed to generate loop invariants automatically, especially on numeric domains \cite{DBLP:conf/cav/Bradley12,DBLP:journals/fac/BradleyM08},
but also in more expressive logics, for programs containing arrays or expressible using combination of theories \cite{DBLP:conf/popl/PadonISKS16,DBLP:conf/vmcai/BeyerHMR07,DBLP:conf/fase/KovacsV09,DBLP:journals/corr/abs-1010-1872,DBLP:journals/jacm/KarbyshevBIRS17,DBLP:conf/cade/KovacsV09}.
We only briefly review 
the main ideas of the most popular and successful approaches.
Methods based on abstract interpretations (see, \eg, \cite{Cousot:1978:ADL:512760.512770,Mine:2006:OAD:1145489.1145526}) 
work by executing the program in a symbolic way, 
on some abstract domain, and try to compute over-estimations of the possible states of the memory after an unbounded number of iterations of the loop.
Counter-examples generated from runs can be exploited to refine the considered abstraction \cite{10.1007/3-540-45251-6_29,10.1007/3-540-44829-2_17}.
The idea is that upon detection of a run for which the assertion is violated, if the run does not correspond to a concrete 
execution path, then the considered abstraction may be refined to dismiss it.

Candidate invariants can also be inferred by generating formulas of some  user-provided patterns and testing them against some  particular
executions of the program \cite{Ernst:1999:DDL:302405.302467}. Those formulas that are violated in any of the runs can be rejected, and the soundness of the remaining candidates can be checked afterwards.
Invariants can be computed by using iterative backward algorithms \cite{SI77}, starting from the post-condition and computing weakest pre-conditions until a fixpoint is reached (if any).
Other approaches \cite{DBLP:journals/jossac/Kapur06} have explored the use of quantifier elimination to refine properties obtained using a representation of all execution paths.

The work that is closest to our approach is \cite{DBLP:conf/oopsla/DilligDLM13}, which presents
an algorithm to compute invariants as boolean combinations of linear constraints over integers. 
The algorithm is 
similar to ours, and also uses abduction to strengthen candidate invariant so that verification conditions are satisfied.
The algorithms differ by the way the verification conditions and abductive hypotheses are proceeded: in our 
approach the conditions always propagate forward from an invariant to another along execution paths, and we eagerly ensure that all the loop invariants are  inductive.
Another difference is that we use a completely different technique to perform abductive reasoning: 
in \cite{DBLP:conf/oopsla/DilligDLM13} is based on model construction and quantifier elimination for Presburger arithmetic, whereas our approach uses a generic algorithm, assuming only the existence of a 
decision procedure for the underlying theory.
This permits to generate invariants expressed in theories that are out of the scope of \cite{DBLP:conf/oopsla/DilligDLM13}.

\paragraph*{Contribution.}
The main contribution is the implementation 
of a general framework for the generation of loop invariants, connecting the platform \whyt and \gpid.
The evaluation demonstrates that the system permits to generate loop invariants for a wide range of theories, though it suffers from a large search space which may induce a large computation time.

\newcommand{\compl}[1]{#1^c}

\section{Verification Conditions}

\newcommand{\posprog}[1]{\mathrm{loc}(#1)}
\newcommand{\looppos}[1]{\mathrm{lloc}(#1)}
\newcommand{\loops}[1]{\mathrm{loops}(#1)}

\newcommand{\ifprog}[3]{\texttt{if } #1 \texttt{ then } #2 \texttt{ else } #3}
\newcommand{\whileprog}[3]{\texttt{while } #1 \texttt{  do } #2 \{ #3 \}  \texttt{ end}}
\newcommand{\forprog}{\texttt{for}}
\newcommand{\aninst}{{\tt I}}
\newcommand{\seqprog}[2]{#1\ ;\ #2}
\newcommand{\assume}[1]{\texttt{assume } #1}
\newcommand{\assert}[1]{\texttt{assert } #1}
\newcommand{\emptyprog}{\texttt{empty}}

\newcommand{\postcondf}{\mathit{sp}}
\newcommand{\postcond}[2]{\postcondf(#1,#2)}
\newcommand{\postcondpos}[3]{\postcondf(#1,#2,#3)}

\newcommand{\precondf}{\mathit{wp}}
\newcommand{\precond}[2]{\precondf(#1,#2)}
\newcommand{\precondpos}[3]{\precondf(#1,#2)}

\newcommand{\sameprog}{\sim}

\label{sect:progs}

In what follows, we consider formulas in a base logic expressing  properties of the memory and assume that such formulas are closed under the usual boolean connectives. These formulas are interpreted modulo some theory $\theory$, where $\modelst$ denotes logical entailment w.r.t.\ $\theory$.
The memory is modified by programs, which are sequences of instructions; they are inductively defined as follows: 
\compact{
\begin{eqnarray*}
    \aprog & = & \emptyprog \quad | \quad \seqprog{\aninst}{\aprog'} \\
\aninst & = &  \tuple{\texttt{base-instruction}} \quad | \quad
 \assume{\aformA} \quad | \quad
 \assert{\aformA} \\
    &  & | \quad \ifprog{\acond}{\aprog_1}{\aprog_2} \quad | \quad
  \whileprog{\acond}{\aprog_1}{\aformA}
  \end{eqnarray*}
}
where $\aprog'$, $\aprog_1$ and $\aprog_2$ are programs, $\acond$ is a condition on the state of the memory, $\phi$ is a 
formula  
and $\aninst$ is an instruction.  Assumptions  correspond to formulas that are taken as hypotheses, they are mostly useful to specify pre-conditions. Assertions correspond to formulas
that are to be proved. Base instructions are left unspecified, they depend on the target language and
application domain; they may include, for instance, assignments and pointer redirection.   
The formula $\aformA$ in the \texttt{while} loop is a \emph{candidate loop invariant}, it is meant to hold every time condition $\acond$ is tested. In our setting each candidate loop invariant will be set to $\top$ before invoking {\Ilinva} (except when another formula is provided by, \eg, the user), and the program will iteratively update these formulas.
We assume that conditions contain no instructions, \ie, that the evaluation of these conditions does not affect the memory.
We write $\aprog \sameprog \aprog'$ if programs $\aprog$ and $\aprog'$ are identical up to the loop candidate invariants.

\newcommand{\position}{location\xspace}
\newcommand{\Position}{Location\xspace}
\newcommand{\location}{\position}
\newcommand{\emptypos}{\varepsilon}
\newcommand{\apos}{\ell}
\newcommand{\proginst}[2]{#1|_{#2}}
\newcommand{\lastpos}[2]{\mathrm{end}(#1,#2)}

An example of a program is provided in Figure \ref{ex:prog}. It uses assignments on integers and usual constructors and functions on lists as base instructions. It contains one loop with candidate invariant $\top$ (Line \ref{line:exinv}) and one assertion (Line \ref{line:exassert}).

\begin{wrapfigure}[10]{L}{0.42\textwidth}
{\scriptsize
\setlength{\interspacetitleruled}{0pt}%
\setlength{\algotitleheightrule}{0pt}%
\begin{algorithm}[H]
    \SetKw{Let}{let}
    \SetKw{List}{list}
    \SetKw{Nil}{nil}
    \SetKw{Unknown}{unknown}
    \SetKw{Assert}{assert}
    
    \SetKw{Head}{head}
    \SetKw{Length}{length}
    
    \SetKwFor{While}{while}{do $\{\top\}$}{}
    
    \Let $\texttt{i} \leftarrow 1$ \;
    \Let $\varlist \leftarrow \List(1, \Nil)$ \;
    
    \While{$\Unknown()$ \label{line:exinv}}{
    	$\texttt{i} \leftarrow \texttt{i} + 1$ \;
        $\varlist \leftarrow \List(\texttt{i}, \varlist)$ \;
    }

    \Assert $\Head(\varlist) = \Length(\varlist)$ \;\label{line:exassert}
\end{algorithm}
\caption{{A simple program on lists\label{ex:prog}}}}
\end{wrapfigure}
It contains one loop for which we will generate an invariant.

A \emph{\position} is a finite sequence of natural numbers. The empty \position is denoted by $\emptypos$ and the concatenation of two {\position}s $\apos$ and $\apos'$ is denoted by $\apos.\apos'$. If $\apos$ is a \position and $S$ is a set of {\position}s then 
$\apos.S$ denotes the set $\{ \apos.\apos' \mid \apos' \in S \}$.
The set of {{\position}s} in a program $\aprog$ or in an instruction $\aninst$ is inductively defined as follows:
\begin{itemize}
\item{If $\aprog$ is an empty sequence then $\posprog{\aprog} = \{ 0 \}$.}
\item{If $\aprog = \seqprog{\aninst}{\aprog'}$ then $\posprog{\aprog} = \{ 0 \} \cup 0.\posprog{\aninst} \cup \{ (i+1).p \mid i \in {\Bbb N}, i.p \in \posprog{\aprog'} \}$.}
\item{If $\aninst$ is a base instruction or an assumption/assertion, then 
$\posprog{\aninst} = \emptyset$.}
\item{If $\aninst = \ifprog{\acond}{\aprog_1}{\aprog_2}$ then $\posprog{\aninst} = 1.\posprog{\aprog_1} \cup 2.\posprog{\aprog_2}$.}
\item{If $\aninst = \whileprog{\acond}{\aprog_1}{\aformA}$ then $\posprog{\aninst} = 1.\posprog{\aprog_1}$.}
 \end{itemize}
For instance, a program $\seqprog{\aninst_1}{\aninst_2}$ where $\aninst_1,\aninst_2$ denote base instructions
has three {\position}s: $0$ (beginning of the program), $1$ (between $\aninst_1$ and $\aninst_2$)
and $2$ (end of the program). Note  that there are no {\position}s within an atomic instruction. 
The program in Figure \ref{ex:prog} has eight {\position}s, namely $0$, $1$, $2$, $2.1.0$, $2.1.1$, $2.1.2$, $3$, $4$.
We denote by $\proginst{\aprog}{\apos}$ the instruction occurring just after \position $\apos$ in $\aprog$ (if any):
\begin{itemize}
\item{If $\aprog = \seqprog{\aninst}{\aprog'}$ then $\proginst{\aprog}{0} = \aninst$,
$\proginst{\aprog}{0.\apos} = \proginst{\aninst}{\apos}$ and 
$\proginst{\aprog}{(i+1).\apos} = \proginst{\aprog'}{i.\apos}$.}
\item{If $\aninst = \ifprog{\acond}{\aprog_1}{\aprog_2}$ then $\proginst{\aninst}{1.\apos} = \proginst{\aprog_1}{\apos}$ and $\proginst{\aninst}{2.\apos} = \proginst{\aprog_2}{\apos}$.}
\item{If $\aninst = \whileprog{\acond}{\aprog_1}{\aformA}$ then $\proginst{\aninst}{1.\apos} = \proginst{\aprog_1}{\apos}$.}
 \end{itemize}
 Note that $\apos \mapsto \proginst{\aprog}{\apos}$ is a partial function, 
 since {\position}s denoting the end of a sequence  do not correspond to an instruction.
 We denote by  $\looppos{\aprog}$ the set of {\position}s $\apos$ in $\aprog$ such that 
 $\proginst{\aprog}{\apos}$ is a loop and by 
 $\loops{\aprog}  = \{ \proginst{\aprog}{\apos} \mid \apos \in \looppos{\aprog} \}$ 
 the set of loops occurring in $\aprog$.  For instance, if $\aprog$ denotes the program in Figure \ref{ex:prog}, then $\proginst{\aprog}{1}$ is $\mathbf{let\ } \varlist \leftarrow \mathbf{list}(1, \mathbf{nil})$, and $\looppos{\aprog} = \{ 2 \}$.

We denote by $<$ the usual order on {\location}s: 
$\apos < \apos'$ iff either there exist
 numbers $i,j$ and {\position}s $\apos_1,\apos_2,\apos_3$ such that 
$\apos = \apos_1.i.\apos_2$, 
$\apos = \apos_1.j.\apos_3$ and
$i < j$, or there exists a \position $\apos''$ such that $\apos' = \apos.\apos''$.

\newcommand{\vcgenf}{\mathrm{VCgen}}
\newcommand{\vcgen}[1]{\vcgenf(#1)}
\newcommand{\vcgena}[1]{\vcgenf_\mathrm{a}(#1)}
\newcommand{\vcgenind}[2]{\vcgenf_{\mathrm{ind}}(#1,#2)}
\newcommand{\vcgeninit}[2]{\vcgenf_{\mathrm{init}}(#1,#2)}
  
  We assume the existence of a procedure $\vcgenf$ that, given a program $\aprog$,
  generates a set of {\em verification conditions} for $\aprog$. 
  These verification conditions are formulas of the form $\phi \Rightarrow \psi$, each of which is meant to be valid. Given a program $\aprog$, the set of conditions $\vcgen{\aprog}$ can be decomposed as follows:
  \begin{enumerate}
  \item{{\em Assertion conditions}, which ensure that the assertion formulas hold at the corresponding \position in the program. These conditions also include additional properties to prevent memory access errors, \eg, to verify that the index of an array is within the defined valid range of indexes. The set of assertion conditions for program $\aprog$ is denoted by $\vcgena{\aprog}$.}
  \item{{\em Propagation conditions}, ensuring that loop invariants do propagate. Given a loop $\aloop$ occurring at position $\apos$ in program $\aprog$, we denote by $\vcgenind{\aprog}{\apos}$ the set of assertions ensuring that the loop invariant for $\aloop$ propagates.}
  \item{{\em Loop pre-conditions}, ensuring that the loop invariants hold when the corresponding loop is entered. Given a loop $\aloop$ occurring at position $\apos$ in program $\aprog$, we denote by $\vcgeninit{\aprog}{\apos}$ the set of assertions ensuring that the loop invariant holds before loop $\aloop$ is entered.}
  \end{enumerate}
Thus, 
\(\vcgen{\aprog} = \vcgena{\aprog} \cup \left(\bigcup_{\apos\in \looppos{\aprog}} (\vcgenind{\aprog}{\apos} \cup \vcgeninit{\aprog}{\apos})\right).\)
Such verification conditions are generally defined using standard weakest pre-condition or strongest post-condition calculi (see, \eg, \cite{Dijkstra:1997:DP:550359}), 
where loop invariant are used as under-appro\-xi\-mations.
 Formal definitions are recalled in Figures \ref{fig:wp} and \ref{fig:sp} (the definition for the basic instructions depends on the application language and is thus omitted). 
 For the sake of readability, we assume, by a slight abuse of notation, that the condition $\acond$  is also a
 formula in the base logic.

\begin{figure}[t]
\compact{
\[
\begin{tabular}{rcl}
\hline
$\precond{\aformA}{\emptyprog}$	& 	$=$	& 	$\aformA$ \\
$\precond{\aformA}{\seqprog{\aninst}{\aprog}}$	& 	$=$	& 	$\precond{\precond{\aformA}{\aprog}}{\aninst}$ \\
$\precond{\aformA}{\assume{\aformA'}}$ &	$=$ & $\aformA' \Rightarrow \aformA$ \\
$\precond{\aformA}{\assert{\aformA'}}$ &	$=$ & $\aformA' \wedge \aformA$  \\
$\precond{\aformA}{\ifprog{\acond}{\aprog_1}{\aprog_2}}$ &	$=$ & $\acond \Rightarrow \precond{\aformA}{\aprog_1} \wedge 
   \neg \acond \Rightarrow \precond{\aformA}{\aprog_2}$ \\
$\precond{\aformA}{\whileprog{\acond}{\aprog_1}{\aformB}}$ & $=$ &
$\aformB \wedge  \forall \vec{x}.~ (\aformB \Rightarrow \precond{\aformB}{\aprog_1}) \wedge \forall \vec{x}.~ (\aformB \wedge \neg \acond \Rightarrow \aformA)$ \\ 
\hline
\end{tabular}
\]
The formula in the last line states that the loop invariant holds when the loop is entered, that it propagates and that it entails the formula $\aformA$ . The vector $\vec{x}$ denotes the vector of variables occurring in $\aprog_1$. }
\caption{\compact{A Weakest Precondition Calculus}\label{fig:wp}}
\end{figure}
This permits to define the goal of the paper in a more formal way:
our aim is to define an algorithm that, given a program $\aprog$, constructs a program $\aprog' \sameprog \aprog$ (\ie,  constructs loop invariants for each loop in $\aprog$)
such that $\vcgen{\aprog'}$ only contains valid formulas.
Note that all the loops and invariants must be handled globally since verification conditions depend on one  another.

\begin{figure}[t]
\compact{
\[
\begin{tabular}{rcl}
\hline
$\postcond{\aformA}{\emptyprog}$	& 	$=$	& 	$\aformA$ \\
$\postcond{\aformA}{\seqprog{\aninst}{\aprog'}}$	& 	$=$	& 	$\postcond{\postcond{\aformA}{\aninst}}{\aprog'}$ \\
$\postcond{\aformA}{\assume{\aformA'}}$ &	$=$ & $\aformA \wedge \aformA'$ \\
$\postcond{\aformA}{\assert{\aformA'}}$ &	$=$ & $\aformA$  \\
$\postcond{\aformA}{\ifprog{\acond}{\aprog_1}{\aprog_2}}$ &	$=$ & $\postcond{\aformA \wedge \acond}{\aprog_1} \vee  
    \postcond{\aformA \wedge \neg \acond}{\aprog_2}$ \\
$\postcond{\aformA}{\whileprog{\acond}{\aprog_1}{\aformB}}$ & $=$ & $\aformB \wedge \neg \acond$ \\
\hline
\end{tabular}
\]
$\postcond{\aformA}{\aprog}$ describes the state of the memory after $\aprog$. The conditions corresponding to loops are approximated by 
using the provided loop invariants (the corresponding verification conditions are not stated).}
 \caption{\compact{A Strongest Postcondition Calculus}\label{fig:sp}}
\end{figure}

\section{Abduction}

\label{sect:abd}

\newcommand{\abducibles}{{\cal A}}
\newcommand{\alit}{l}
\newcommand{\acl}{I}
\newcommand{\calP}{{\cal P}}
\newcommand{\cmodel}{\mathfrak{m}}
\newcommand{\simplify}{\textsc{Simplify}}
\newcommand{\impid}{\textsc{GPiD}}
\newcommand{\set}[1]{\left\{#1\right\}}
\newcommand{\UNSAT}[1]{$#1$ unsatisfiable (modulo $\theory$)}
 \newcommand{\implicant}{$\abducibles$-implicant\xspace}
 
As mentioned above, abductive reasoning will be performed by generating implicants. Because it would not be efficient to blindly generate all implicants of a formula, this generation is controlled by fixing the literals that can occur in an implicant. We thus consider a set $\abducibles$ of literals in the considered logic, called the {\em abducible literals}.

\compact{
\begin{algorithm}[t]
	\caption{\impid($\aformA, M, A, \calP$)}\label{alg:impid}
	\SetKw{Let}{let}
	\SetKw{Simplify}{simplify}
	\SetKw{And}{and}
	\SetKw{Or}{or}
	\If{\UNSAT{M} \Or $\neg \calP(M)$ \label{line:sat}}
	{
		\Return $\emptyset$\;	
	}
	\If{$M \models \aformA$}
	{		
		\Return $\set{ M }$\;
	}
	\Let{$\cmodel$ be a model of $\{ \neg \aformA \}  \cup M$\;}
	\Let{$\aformA = \simplify(\aformA,M)$\;}
	\Let{$A = \{ \alit \in A \mid M \cup \neg \aformA \not \modelst \alit, M  \not \modelst \compl{\alit} \}$\;}  \label{line:linkedsimp}
	\ForEach{$\alit\in A$ such that $\cmodel \not \models \alit$\label{line:mod}}
	{		
		\Let{$A_\alit = \setof{\alit' \in A}{\alit' < \alit \wedge \cmodel \models \alit'} \cup \setof{\alit' \in A}{\alit < \alit'}$\;\label{line:Alit}}
		\Let{$P_\alit = \impid(\aformA, M\cup \set{\alit}, A_\alit, \calP)$\;\label{line:Reccal}}		
	}
	\Return $\bigcup_{\alit\in A} P_\alit$\;	
\end{algorithm}}

\begin{definition}
Let $\aformA$ be a formula.
An {\em \implicant} of $\aformA$ (modulo $\theory$) is a conjunction
(or set) of literals $\alit_1 \wedge \dots \wedge \alit_n$ such that
$\alit_i \in \abducibles$, for all $i \in \interv{1}{n}$ and 
$\alit_1 \wedge \dots \wedge \alit_n \modelst \aformA$.
\end{definition}
We use the procedure \gpid described in \cite{EPS18} to generate 
 {\implicant}s.
A simplified version of this procedure is presented in Algorithm \ref{alg:impid}.
A call to the procedure \impid($\aformA, M, A, \calP$) is meant to generate {\implicant}s of $\phi$ that: (i) are of the form $M \cup A'$, for some $A' \subseteq A$; (ii) are as general as possible; and (iii) satisfy property $\calP$. When $M$ itself is not an {\implicant} of $\phi$, a subset of relevant literals from $A$ is computed (Line \ref{line:linkedsimp}), and for each literal in this subset, a recursive call is made to the procedure after augmenting $M$ with this literal and discarding all those that become irrelevant (Lines \ref{line:Alit} and \ref{line:Reccal}). In particular, the algorithm is parameterized by an ordering $<$ on abducible literals which is used
to ensure that sets of hypotheses are explored in a non-redundant way. The algorithm relies on the existence of a decision procedure for testing satisfiability in $\theory$ (Line \ref{line:sat}).
In practice, this procedure does not need terminating or complete\footnote{However, Theorem \ref{theo:sound} only holds if the proof procedure is terminating and complete.}, \eg, it may be called  with a timeout (any ``unknown'' result is handled as ``satisfiable'').
At Line \ref{line:mod}, a model of the formula
$\{ \neg \aformA \}  \cup M$ is used to prune the search space, by dismissing some
abducible literals. In practice, no such model may be available,
either because no model building algorithm exists for the considered theory or because of termination issues. In this case, no such pruning is performed.
Property $\calP$ denotes an abstract property of sets of literals. It is used to control the form of generated {\implicant}s, it is for example possible to force the algorithm to only generate {\implicant}s with a fixed maximal size. For Theorem \ref{theo:sound} to hold, it is simply required that $\calP$ be \emph{closed under subsets}, \ie, that for all sets of abducible literals $B$ and $C$, $B \subseteq C \wedge \calP(C) \Rightarrow \calP(B)$. 

Compared to \cite{EPS18}, details that are irrelevant for the purpose of the present paper are skipped and the procedure has been adapted to generate {\implicant}s instead of implicates (implicants and implicates are dual notions).

\begin{theorem}[\cite{EPS18}]
\label{theo:sound}
The call $\impid(\aformA, 
\emptyset, \abducibles, \calP)$ terminates and returns a set of 
{\implicant}s of $\aformA$  satisfying $\calP$. Further, if $\calP$ is closed under subsets, then
for every \implicant $\acl$ of $\aformA$ satisfying $\calP$, there exists $\acl' \in \impid(\aformA, 
\emptyset, \abducibles, \calP)$ such that $\acl \modelst \acl'$.
\end{theorem}
This procedure  also comes with generic algorithms for 
pruning redundant {\implicant}s 
\ie, for removing all {\implicant}s $\acl$ such that there exist another {\implicant} $\acl'$ such that 
$\acl \modelst \acl'$, see \cite[Section $4$]{EPS18}.

\section{Generating Loop Invariants}

\label{sect:gen}

\newcommand{\hypsA}{\Xi}
\newcommand{\hypA}{\xi}

\newcommand{\backprop}[4]{\mathit{bp}(#1,#2,#3,#4)}
\newcommand{\forwardprop}[4]{\mathit{fp}(#1,#2,#3,#4)}
\newcommand{\extprog}[3]{\mathit{path}(#1,#2,#3)}
\newcommand{\nowhile}[1]{\mathit{RmLoops}(#1)}

\newcommand{\str}[3]{\mathrm{Strengthen}(#1,#2,#3)}
\newcommand{\fail}{\mathbf{fail}}
\newcommand{\extractabd}{\textsc{GetAbducibles}}
\newcommand{\Abduce}{\textsc{Abduce}}

In this section, we present an algorithm for the generation of loop invariants.
As explained in Section \ref{sect:progs}, we distinguish between $3$ kinds of verification conditions, which will be handled in different ways: assertion and propagation conditions; and loop pre-conditions.
As can be seen from the rules in Figure \ref{fig:wp}, 
loop invariants can occur as antecedents in verification conditions, this is typically the case
when a loop occurs just before an assertion in some execution path.
In such a situation, we say that the considered condition {\em depends on}  loop $\aloop$.
When a condition depends on a loop, a strengthening of the loop invariant of loop $\aloop$ yields a 
strengthening of the hypotheses of the verification condition, \ie, makes the condition 
less general (easier to prove).

\compact{
\begin{algorithm}[t]
	\caption{\Ilinva(Program $\aprog$)}
	\label{algo:ilinva}
	\SetKw{Let}{let}
	\SetKw{SuchThat}{such that}
	\SetKw{Break}{break}
	\If{all formulas in $\vcgena{\aprog}$ are valid\label{line:test}}{
		\Return $\aprog$\;
	}
	\Let $\aformA$ be a non valid formula in $\vcgena{\aprog}$, depending on a loop at \location $\apos$\;  \label{line:select}
	\Let $\hypsA \longleftarrow \Abduce(\aformA,\aprog,\apos)$\;  \label{line:abd}
	\ForEach{$\hypA \in \hypsA$\label{line:foreach}}{ 
		\ForEach{$\apos'\in \looppos{\aprog}$ \SuchThat $\apos'\leq \apos$ \label{line:iterloop}}{
			\Let $\hypA' \longleftarrow \backprop{\hypA}{\aprog}{\apos}{\apos'}$\; \label{line:bk} 
			\Let $\aprog_{\hypA} \longleftarrow \str{\aprog}{\apos'}{\hypA'}$ \;\label{line:str}
			\If{$\vcgeninit{\aprog_{\hypA}}{\apos'}$ is valid\label{line:filter}}{
				\Let $\aprog_{\hypA}' \longleftarrow \Propagate(\aprog_{\hypA}, \apos')$ \label{line:prog} \;
				\If{$\aprog_{\hypA}' \not = \fail$}{
					\Let $\aprog_{\hypA}'' \longleftarrow \Ilinva(\aprog_{\hypA}')$\;\label{line:rec}
					\If{$\aprog_{\hypA}'' \not = \fail$}{\Return $\aprog_{\hypA}''$\;}
				} 
			}
		}
	}
	
	\Return{$\fail$} \;  
\end{algorithm}}

This principle is used in Algorithm \ref{algo:ilinva}, which we briefly describe before going into details. Starting with a program $\aprog$ in which it is assumed that every loop invariant is inductive, the algorithm attempts to recursively generate invariants that make all assertion conditions in $\aprog$ valid. It begins by selecting a non-valid formula $\aformA$ from $\vcgena{\aprog}$ and a {\location} $\apos\in \looppos{\aprog}$ such that $\aformA$ depends on $\apos$, then generates a set of hypotheses that would make $\phi$ valid (Line \ref{line:abd}). For each such hypothesis $\hypA$, a loop {\location} $\apos'$ such that $\apos'\leq\apos$ is selected, and a formula $\hypA'$ that is a weakest precondition at $\apos'$ causing $\hypA$ to hold at {\location} $\apos$ is computed (Line \ref{line:bk}). This formula is added to the invariant of the loop at {\location} $\apos'$ (Line \ref{line:str}), so that if this invariant was $\psi$, the new candidate invariant is $\hypA'\wedge \psi$. If $\hypA'$ does not hold before entering the loop then $\hypA$ is discarded (Line \ref{line:filter}); otherwise, the program attempts to update the other loop invariants to ensure that $\hypA'$ propagates (Line \ref{line:prog}). When this succeeds, a recursive call is made with the updated invariants (Line \ref{line:rec}) to handle the other non-valid assertion conditions.

\compact{
\begin{algorithm}[t]
	\caption{\Abduce(Formula $\aformA$, Program $\aprog$, \Position $\apos$)}
	\label{algo:abduce}
	\SetKw{Let}{let}
	\SetKw{SuchThat}{such that}
	\SetKw{Break}{break}
	\Let $\abducibles \longleftarrow \extractabd(\aformA)$ \label{line:collect} \; 
	\Let $\abducibles \longleftarrow \{ \alit \mid \alit \in \abducibles \wedge \aformA \not \modelst \alit \}$ \label{line:filtabd} \; 
	\Let $\hypsA \longleftarrow \impid(\aformA,\emptyset, \abducibles, {\calP})) $ \; 
	\Let $\hypsA' \longleftarrow \{ \hypA_1 \vee \dots \vee \hypA_n \mid n \in {\Bbb N}, \hypA_i \in \hypsA \}$ \label{line:disj} \; 
	\Return $\hypsA'$ \\
\end{algorithm}}

Procedure $\Abduce(\aformA)$ (invoked Line \ref{line:abd} of Algorithm \ref{algo:ilinva}) is described in Algorithm \ref{algo:abduce}. It generates  formulas $\hypA$ that logically entail $\aformA$; it is used to generate the candidate hypotheses for strengthening.
It first extracts a set of abducible literals $\abducibles$ by 
collecting variables and symbols from the program and/or from the theory $\theory$ and combining them to create literals up to a certain depth (procedure $\extractabd$ at Line \ref{line:collect}).
To avoid any redundancy, this task is actually done in two steps: a set of abducible literals for the entire program is initially constructed (this is done once at the beginning of the search), and depending on the considered program \position, a subset of these literals is selected.
The abducible literals that are logically entailed by $\aformA$ modulo $\theory$ are filtered out (Line \ref{line:filtabd}), and  procedure $\impid$ is called to generate {\implicant}s of $\aformA$.
Finally, {\implicant}s are combined to form disjunctive formulas.
Note that another way of generating disjunction of literals would be to add these disjunction in the initial set of abducible literals, but this solution would greatly increase the search space.

Each of the hypotheses $\hypA$ generated by $\Abduce(\aformA)$ is used to strengthen the invariant of a loop occurring at position $\apos'\leq \apos$
(Line \ref{line:str} in Algorithm \ref{algo:ilinva}). The strengthening formula is computed using the Weakest Precondition Calculus on $\hypA$, on a program obtained from $\aprog$ by ignoring all loops between $\apos'$ and $\apos$, since they have corresponding invariants.
To this purpose we define a function $\backprop{\aformA}{\aprog}{\apos}{\apos'}$ which, for positions $\apos' \leq \apos$, back-propagates 
abductive hypotheses from a \location $\apos$ to $\apos'$ (see Figure \ref{fig:bk}). 
This is done by extracting the part of the code 
$\extprog{\aprog}{\apos'}{\apos}$ between the {\location}s $\apos'$ and $\apos$  while ignoring loops, and computing the weakest precondition  corresponding to this part of the code and the formula $\aformA$.

\begin{figure}[t]
\compact{
\[
\begin{tabular}{rcll}
 $\extprog{\aprog}{\apos}{\apos}$	& 	$=$	& 	$\emptyprog$ \\
 $\extprog{\aprog}{\apos}{\apos'.(i+1)}$	& 	$=$	& 	$\extprog{\aprog}{\apos}{\apos'.i} \bullet \proginst{\aprog}{\apos'.i}$  &\ if $\apos \leq \apos'.i$\\
 $\extprog{\aprog}{\apos}{\apos'.0}$ & $=$	& $\extprog{\aprog}{\apos}{\apos'}$ &\ if $\apos \leq \apos'$ \\\vspace{1em}
 $\extprog{\aprog}{\apos.i.\apos'}{\apos.{(i+1)}}$ &$=$&  $\extprog{\aprog}{\apos.i.\apos'}{\apos.i.{m}}$ &\ $m = \max \{ j \mid \apos.i.j \in \posprog{\aprog} \}$ \\ 
 $\backprop{\aformA}{\aprog}{\apos}{\apos'}$	& 	$=$	& 	
 $\precond{\aformA}{\aprog'}$ &\ if 
 $\aprog' = \extprog{\nowhile{\aprog}}{\apos'}{\apos}$ \\
 $\forwardprop{\aformA}{\aprog}{\apos}{\apos'}$	& 	$=$	& 	
 $\postcond{\aformA}{\aprog'}$ &\ if 
 $\aprog' = \extprog{\nowhile{\aprog}}{\apos}{\apos'}$ \\

 \end{tabular}
 \]
 $\nowhile{\aprog}$ denotes the program obtained from $\aprog$ by removing all \texttt{while} instructions and $\bullet$ denotes the concatenation operator on programs.}
\caption{\compact{Backward and Forward Propagation of Abductive Hypotheses} \label{fig:bk}}
\end{figure}

\compact{
\begin{algorithm}[t]
	\caption{\Propagate(Program $\aprog$, \Position $\apos$)}
	\label{algo:prop}
	\SetKw{Let}{let}
	\SetKw{SuchThat}{such that}
	\SetKw{Break}{break}
	\If{all formulas in $\vcgenind{\aprog}{\apos}$ are valid}{
		\Return $\aprog$\;
	}
	\Let $\aformA$ be a non-valid formula in $\vcgenind{\aprog}{\apos}$ \;  
	\Let $\hypsA \longleftarrow \Abduce(\aformA,\aprog,\apos)$\; 
	\ForEach{$\hypA \in \hypsA$}{ 
		\ForEach{$\apos' \in \looppos{\aprog}$ such that $\apos$ is a prefix  of $\apos'$ (with possibly $\apos = \apos'$)}{
			\Let $\hypA' \longleftarrow \forwardprop{\hypA}{\aprog}{\apos}{\apos'}$\; 
			\Let $\aprog_{\hypA}' \longleftarrow \str{\aprog}{\apos'}{\hypA'}$ \;
			\If{$\vcgeninit{\aprog_{\hypA}'}{\apos'}$ is valid\label{line:filterprop}}{
				\Let $\aprog_{\hypA}'' \longleftarrow \Propagate(\aprog_{\hypA}', \apos)$ \;
				\If{$\aprog_{\hypA}'' \not = \fail$}{
					\Return $\aprog_{\hypA}''$\;}
			}
		} 
	}
	\Return{$\fail$} \;  
\end{algorithm}}
The addition of hypothesis $\hypA'$ to the invariant of the loop at position $\apos'$ ensures that the considered assertion $\aformA$ holds, 
but it is necessary to ensure that this strengthened invariant is still inductive.
This is done as follows.
Line \ref{line:filter} of Algorithm \ref{algo:ilinva} filters away all candidates for which the precondition before entering the loop is no longer valid, and Algorithm \ref{algo:prop} ensures that the candidate still propagates.
This algorithm behaves similarly to Algorithm 
\ref{algo:ilinva} (testing the verification conditions in $\vcgenind{\aprog}{\apos}$ instead of  those
in $\vcgena{\aprog}$), except that it  strengthens the invariants that correspond either to 
the considered loop, or to other loops occurring within it (in the case of nested loops).
Note that in this case, properties must be propagated forward, from \position $\apos$ to the actual 
\position of the strengthened invariant, using a Strongest Postcondition Calculus (Function $\forwardprop{\aformA}{\aprog}{\apos}{\apos'}$ in Figure \ref{fig:bk}).
This technique avoids considering  hypotheses that do not  propagate.

When applied on the program in Figure \ref{ex:prog},
\Ilinva  first sets the initial invariant of the loop to $\true$ and 
considers the assertion 
$\phi: \mathbf{head}(\varlist) = \mathbf{length}(\varlist)$. 
As the entailment 
$\true \models \phi$ does not hold, 
it will call \gpid to get an implicant of $\true \Rightarrow \phi$. Assume that \gpid returns the (trivial)
solution $\phi$. As $\phi$ indeed holds when the loop is entered\footnote{This can be checked by 
computing the weakest precondition
of $\phi$ w.r.t.\ Lines $1,2$. The obtained formula is 
$\mathbf{head}(\mathbf{list}(1,\mathbf{nil})) = \mathbf{length}(\mathbf{list}(1,\mathbf{nil}))$ which is equivalent to $\true$ 
(w.r.t.\ the usual definitions of $\mathbf{list}$ and $\mathbf{head}$).}, \Ilinva will add $\phi$ to the invariant of the loop and call \Propagate.  Since $\phi$ does not propagate \Propagate will further strengthen the invariant, yielding, \eg, the correct solution:
$\phi \wedge \texttt{i}  = \mathbf{head}(\varlist)$. 

The efficiency of Algorithm \ref{algo:ilinva} crucially depends on the order 
in which  candidate hypotheses are processed at Line \ref{line:foreach} for the strengthening operation.
The  heuristic used in our current implementation is to try the simplest hypotheses with the highest priority.
Abducible atoms are therefore ordered as follows: first boolean variables, then equations between variables of the same sort, 
then applications of predicate symbols to variables (of the appropriate sorts) 
and finally deep literals involving function symbols (up to a certain depth).
In every case, negative literals are also considered, with the same priority as the corresponding atom.
Similarly,  
unit {\implicant}s are tested before non-unit ones, and single {\implicant}s before disjunctions of {\implicant}s.
In the iteration on line \ref{line:iterloop} of Algorithm \ref{algo:ilinva}, the loops that are closest to the considered assertions are considered first.
Due to the number of loops involved,
numerous parameters are used to control the application of the procedures, by fixing limits
on the number of abducible literals that may be considered and on the maximal size of {\implicant}s. 
When a call to $\Ilinva$ fails, these parameters are increased, using an iterative deepening search strategy.
The parameter controlling the maximal number of {\implicant}s in the disjunctions (currently either $1$ or $2$) is fixed outside of the loop as it has a strong impact on the computation cost. 

The following theorem states the main properties of the algorithm.
\begin{theorem}
\label{theo:ilinva}
Let $\aprog$ be a program such that 
$\vcgenind{\aprog}{\apos}$ and 
$\vcgeninit{\aprog}{\apos}$ are valid for all $\apos \in \looppos{\aprog}$.
If \Ilinva($\aprog$) terminates and returns a program $\aprog'$ other than $\fail$,  then 
$\aprog \sameprog \aprog'$ 
and $\vcgen{\aprog'}$ is valid modulo $\theory$.
Furthermore, if the considered set of abducible literals is finite (\ie, if there exists a 
finite set $\abducibles$ such that $\extractabd(\aformA) \subseteq \abducibles$ for all formulas 
$\aformA$), then \Ilinva($\aprog$) terminates.
\end{theorem}
\begin{proof}
\begin{conference}
    The proof is provided in the extended version\footnote{\todo{Import the extended version on arxiv and add the link to the extended version here}}.
\end{conference}
\begin{report}
The proof is by induction on the recursive calls.
It is clear that $\aprog \sameprog \aprog'$ because the algorithm only modifies loop invariants.
By construction (Line \ref{line:test} of Algorithm \ref{algo:ilinva}), 
$\vcgena{\aprog}$ must be valid when $\aprog$ is returned.
By hypothesis $\vcgenind{\aprog}{\apos}$ and 
$\vcgeninit{\aprog}{\apos}$ are valid, and 
by definition $\vcgen{\aprog}  = \vcgena{\aprog} \cup \bigcup_{\apos\in \looppos{\aprog}} (\vcgeninit{\aprog}{\apos} \cup \vcgenind{\aprog}{\apos})$, thus $\vcgen{\aprog}$ is valid in this case.
Furthermore, it is easy to check, by inspection of 
Algorithm \ref{algo:prop}, that 
all the recursive calls to \Ilinva occur on programs such that 
$\vcgenind{\aprog}{\apos}$  is valid. 
Furthermore, all the formulas
$\vcgeninit{\aprog}{\apos'}$ are also valid, due to the tests at Line \ref{line:filter} 
in Algorithm \ref{algo:ilinva} and Line \ref{line:filterprop} in Algorithm \ref{algo:prop} (indeed, it is clear that the 
strengthening of the invariant at \position $\apos'$ preserves the validity of $\vcgeninit{\aprog}{\apos}$ for $\apos \not = \apos'$).
Thus the precondition above holds for these recursive calls and the result follows by the induction hypothesis.

Termination is immediate since there are only finitely many possible candidate invariants built on $\abducibles$, thus
the (strict) strengthening relation (formally defined as: $\aformA > \aformB \iff \aformB \models \aformA \wedge \aformA \not \models \aformB$) forms a well-founded order.
At each recursive call, one of the invariants is strictly strengthened and the other ones are left unchanged hence
the multiset of invariants is strictly decreasing, according to the multiset extension of the  strengthening relation.
\end{report}

\end{proof}

\section{Implementation}

\subsection{Overview}

The \Ilinva algorithm described in Section \ref{sect:gen} has been implemented by connecting \why 
with \gpid.
\begin{report}
        A workflow graph of this implementation is shown on Figure \ref{fig:workflow}.
        The input file (a \whyml program) and a configuration is forwarded to the tool via the command line.
        The tool then loads this input file within a wrapper where it identifies the candidate loop invariants that may be strengthened by the system.
        This wrapper will also modify the candidate loop invariants within the program when strengthened, and can export the corresponding file at any time.
        
        The main system then forwards the configuration to the invariant generation algorithm.
        During the execution of the \ilinva algorithm, the \whyml wrapper is tasked to query \whyt to check whether the latter is able to prove all the assertions of the updated program, and if not to recover the verification conditions that are not satisfied.
        Selected conditions are transferred to the main generator which will create appropriate abduction tasks for them, ask \gpid for implicants, select the meaningful ones and strengthen associated candidate loop invariants accordingly.
        When a proof for the verification conditions of the program is found, the file wrapper returns the program updated with the corresponding loop invariants.
        We also expressed that \gpid candidates can be pruned when they contradict loops initial conditions.
        Note that both \gpid and \whyt call external SMT solvers to check the satisfiability of formulas.
\end{report}
\begin{conference}
    A workflow graph of this implementation is detailed in the extended paper.
    Note that both systems themselves call external SMT solvers to check the satisfiability of formulas.
\end{conference}
In particular, the \gpid toolbox
is easy to plug to any \smtlib-compliant SMT solver. 
The framework is actually generic, in the sense that it could be plugged with other systems, both to generate and verify proof obligations and to strengthen loop invariants.
 It is also independent of the constructions used for defining the language: other 
constructions (\eg, \forprog\ loops) can be considered, provided they are handled by the program verification framework.

\begin{report}
    \begin{figure}
        \begin{center}
            \begin{tikzpicture}[scale=0.09]

\draw (80,0) -- (100,0) -- (100,30);
\draw (100,70) -- (100,100) -- (0,100) -- (0,0) -- (12.5,0);
\draw (17.5,0) -- (40,0);

\draw (40,-2.5) rectangle (80,2.5);
\node at (60,0) {problem data interface};

\draw (97.5,30) rectangle (102.5,70);
\node [rotate=-90] at (100,50) {result interface};

\draw (15,0) circle (2.5);
\node at (15,0) {\texttt{K}};

\draw (40,-5) rectangle (80,-10);
\node at (60,-7.5) {\ilinva-generator};
\draw (40,-15) rectangle (80,-10);
\node at (60,-12.5) {\ilinva-executable};

\draw (5,10) rectangle (55,55);
\node at (20,12.5) {\textbf{Inv. Generator}};
\draw [pattern=north east lines] (40,8.5) rectangle (55,10);
\draw [pattern=north east lines] (55,42.5) rectangle (56.5,52.5);
\draw [pattern=north east lines] (55,27.5) rectangle (56.5,37.5);
\draw [pattern=north east lines] (55,12.5) rectangle (56.5,22.5);

\draw (30,30) rectangle (50,50);
\node at (40,37.5) {\scriptsize abduction};
\node at (40,35) {\scriptsize problem};
\node at (40,32.5) {\scriptsize generator};
\draw [pattern=north east lines] (30,55) rectangle (50,56.5);
\draw (10,20) rectangle (15,50);
\node [rotate=90] at (12.5,35) {\footnotesize strengths. pool};
\draw (20,20) rectangle (25,50);
\node [rotate=90] at (22.5,35) {\footnotesize updated pool};
\draw [pattern=north east lines] (10,55) rectangle (15,56.5);
\draw (30,15) rectangle (50,20);
\node at (40,17.5) {\scriptsize stren. selector};

\draw (65,10) rectangle (95,95);
\node at (80,12.5) {\scriptsize\textbf{\whyml Wrapper}};
\draw [pattern=north east lines] (95,30) rectangle (96.5,70);
\draw [pattern=north east lines] (65,8.5) rectangle (80,10);
\draw [pattern=north east lines] (63.5,77.5) rectangle (65,87.5);
\draw [pattern=north east lines] (63.5,42.5) rectangle (65,52.5);
\draw [pattern=north east lines] (63.5,27.5) rectangle (65,37.5);
\draw [pattern=north east lines] (63.5,12.5) rectangle (65,22.5);

\draw (70,75) rectangle (90,90);
\draw [pattern=north west lines, draw=none] (70,75) rectangle (90,77.5);
\draw [pattern=north west lines, draw=none] (70,87.5) rectangle (90,90);
\draw [pattern=north west lines, draw=none] (70,77.5) rectangle (72.5,87.5);
\draw [pattern=north west lines, draw=none] (87.5,77.5) rectangle (90,87.5);
\draw (72.5,77.5) rectangle (87.5,87.5);
\node at (80,82.5) {\textbf{\whyt}};

\draw (5,70) rectangle (55,95);
\node at (30,82.5) {\textbf{\gpid}};
\draw [pattern=north east lines] (55,77.5) rectangle (56.5,87.5);
\draw [pattern=north east lines] (30,68.5) rectangle (50,70);
\draw [pattern=north east lines] (3.5,77.5) rectangle (5,87.5);

\draw (55, -17.5) node [left] {\texttt{Input File + Opts}} -- (60,-17.5) -- (60,-15);
\draw [dotted] (60,-15) -- (60,-5);
\draw [->,>=latex] (60,-5) -- (60,-2.5);

\draw [->,>=latex] (15,-10) -- (15,-2.5);

\draw [->,>=latex] (72.5,2.5) -- node [rotate=90,above] {\texttt{\scriptsize file}} (72.5,8.5);
\draw [->,>=latex] (47.5,2.5) -- node [rotate=90,below] {\texttt{\scriptsize opts}} (47.5,8.5);

\draw (92.5,50) -- (97.5,50);
\draw [dotted] (97.5,50) -- (102.5,50);
\draw [->,>=latex] (102.5,50) -- (105,50) -- (105,-17.5) -- (95,-17.5) node [left] {\texttt{Output File}};

\draw (56.5,50) -- node [below] {\emph{\tiny queries}} (63.5,50) -- (65,50);
\draw [dotted,->,>=latex] (65,50) -- (77.5,50) -- (77.5,75);
\draw [dotted] (65,45) -- (82.5,45) -- (82.5,75);
\draw [<-,>=latex] (56.5,45) -- (65,45);

\draw [dotted,<-,>=latex] (50,32.5) -- (55,32.5);
\draw (55,32.5) -- node [above] {\emph{\scriptsize VC}} (65,32.5);
\draw [dotted,->,>=latex] (82.5,45) -- (82.5,32.5) -- (80,32.5);
\draw [dotted] (70,32.5) -- (65,32.5);
\node at (75,33) {\scriptsize unproven};
\node at (75,31) {\scriptsize VCs};
\draw [dotted] (50,17.5) -- (55,17.5);
\draw (55,17.5) -- (65,17.5);
\draw [dotted,->,>=latex] (65,17.5) -- node [above] {\emph{\tiny strengthen}} (77.5,17.5) ;

\node at (85,21) {\scriptsize Candidate};
\node at (85,19) {\scriptsize loop};
\node at (85,17) {\scriptsize invariants};

\draw [dotted,->,>=latex] (22.5,20) -- (22.5,17.5) -- (30,17.5);

\draw [dotted,->,>=latex] (15,47.5) -- (20,47.5);
\draw [dotted,->,>=latex] (15,42.5) -- (20,42.5);
\draw [dotted,->,>=latex] (15,27.5) -- (20,27.5);
\draw [dotted,->,>=latex] (15,22.5) -- (20,22.5);
\node [rotate=90] at (17.5,35) {\emph{\scriptsize prune/disjunct}};

\draw [dotted,->,>=latex] (42.5,10) -- node [rotate=90,above] {\emph{\tiny uses}} (42.5,15);
\draw [dotted] (47.5,10) -- (47.5,15);
\draw [loosely dotted] (47.5,15) -- (47.5,20);
\draw [dotted,->,>=latex] (47.5,20) -- node [rotate=90,below] {\emph{\tiny uses}} (47.5,30);

\draw [dotted] (40,50) -- (40,55);
\draw [->,>=latex] (40,55) -- (40,68.5);

\draw [dashed,->,>=latex] (45,56.5) -- (45,60) -- (60,60) -- node [rotate=90,above] {\emph{\tiny configures}} (60,82.5);
\draw [<-,>=latex] (56.5,82.5) -- node [above] {\emph{\tiny $\mathcal{P}_{\text{init}}$}} (63.5,82.5);

\draw (3.5,82.5) -- (2.5,82.5) -- (2.5,60) -- (12.5,60) -- (12.5,55);
\draw [dotted,->,>=latex] (12.5,55) -- (12.5,50);

\draw [loosely dashed,->,>=latex] (15,2.5) -- (15,10);
\draw [loosely dashed,->,>=latex] (15,5) -- (2.5,5) -- (2.5,57.5) -- (7.5,57.5) -- (7.5,70);

\end{tikzpicture}
        \end{center}
        \caption{Workflow graph of the \ilinva tool\label{fig:workflow}}
    \end{figure}
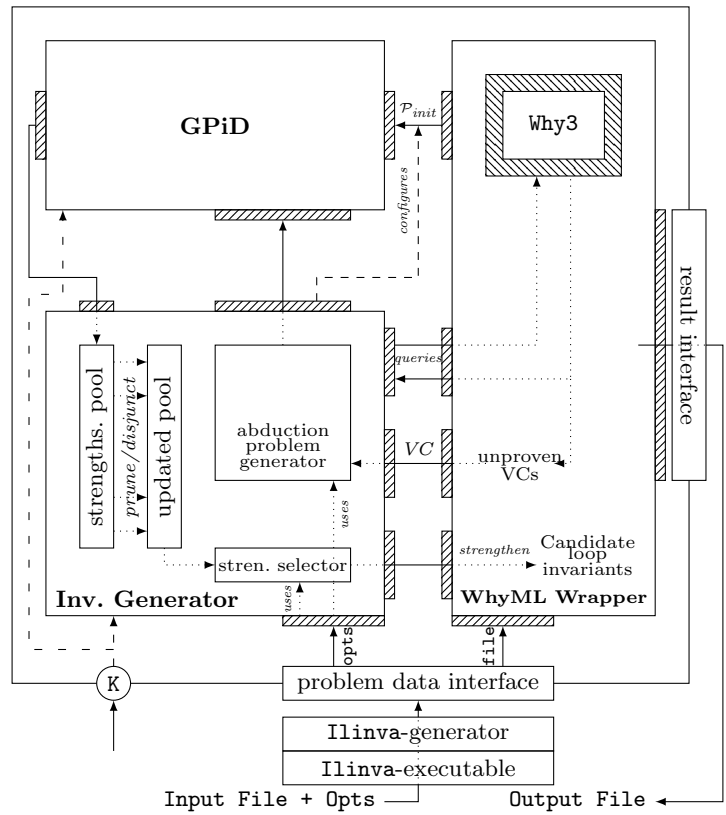
\end{report}

Given an input program written in \whyml, \whyt generates a verification condition the validity of which will ensure that all the asserted properties are verified (including additional conditions related to, \eg, memory safety)
This initial verification condition is split by \whyt into several subtasks. These conditions are enriched with all appropriate information (\eg, theories, axioms,\dots) and sent to external SMT solvers to check satisfiability.
The conditions we are interested in are those linked to the proofs of the program assertions, as well as those ensuring that the candidate loop invariants are inductive.
In our implementation, \why is taken as a black box, and we merely recover the files
that are passed from \whyt to the SMT solvers,  together with additional configuration  data for the solvers we can extract from \whyt.
If the proof obligation fails, then we relate the file to the corresponding assertion in the \whyml  program and extract the set of abducible literals as explained
in Section \ref{sect:gen}, restricting ourselves to symbols 
corresponding to \whyml variables, functions and predicates.
We then  tune the \smtlib file to adapt it for computations by \gpid
and invoke \gpid 
with the same SMT-solver as the one used by \whyt to check satisfiability, as the problem is expected to be  optimized/configured for it.
We also configure \gpid to skip the exploration of subtrees that will produce candidate invariants that do not satisfy the loop preconditions.
\gpid returns a stream of solutions to the abductive reasoning problem. We then backward-translate the formulas into the \whyml formalism and use them to strengthen loop invariants.
For efficiency, the systems run in parallel:  
the generation of abductive hypotheses (by $\gpid$, via the procedure $\Abduce$) 
and their processing in $\whyml$ (via $\Ilinva$) is organized as a pipe-line, where new abduction solutions are computed 
during the processing of the first ones. 

To bridge \Ilinva and \whyt, we had to devise an interface, which
is able to analyze
\whyml programs and to identify loop {\location}s and the corresponding invariants.
It invokes \whyt to generate and prove the associated verification tasks, and it recovers the failed ones.
The library also includes tools to extract and modify loop invariants,
 to extract variables and reference variables in \whyml files, as well as
types, predicates and functions, and wrappers to call the \whyt executable and external tools, and to  extract the files sent by  \whyml to SMT-solvers.

\subsection{Distribution}
The \abdulot framework is available on \github \cite{gpidGithub}.
It contains an revamped interface to the \gpid libraries and algorithm, a generic library of the \ilinva algorithm automatically plugged with \gpid, the code interface for \whyt and the related executables.
\gpid interfaces and related executables are generated for \cvcf, \zthree and \altergo\footnote{Those are the three solvers the \whyt documentation recommends to work with as an initial setup. (see also \url{http://why3.lri.fr/}@External Provers.)} via their \smtlib interface. Note that the SMT solvers are not provided by our framework, they must be installed separately (all versions with an \smtlib-compatible interface are supported). 
Additional interfaces and executables can be produced using \cpp libraries for \minisat, \cvcf and \zthree if their supported version is available\footnote{The \altergo interface provided by the tool uses an \smtlib interface that is under heavy development and that, in practice, does not work well with the examples we send it.}.

The framework also provides libraries and toolbox executables to work with abducible files, \cpp libraries to handle \whyml files,  helpers for the generation of abducible literals out of \smtlib files, and an extensive lisp parser.
It also includes a documentation, which explains in particular how to extend it to other solvers and program verification framework.
All the tools can be compiled using any \cpp 11 compliant compiler.
The whole list of dependencies is available in the documentation, as well as a dependency graph for the different parts of the framework.

\section{Experiments}

We use benchmarks collected from several sources \cite{DBLP:conf/oopsla/DilligDLM13,InvgenTool,NeclaBenchmark,SatConv,Why3Examples0,Why3Examples2,Why3ExamplesF} (see also \cite{gpidGithub} for a more detailed view of the benchmark sources),
 with additional examples corresponding
to standard algorithms for division and exponentiation (involving lists, arrays, 
and non linear arithmetic). 
Some of these benchmarks have been  translated\footnote{The translation was done by hand.} from \ccode or \java into \whyml.
In all cases, the initial invariant associated with each loop is $\true$.
We used \zthree for the benchmarks containing real arithmetic, \altergo for lists and arrays and \cvcf in all the other cases.
All examples are embedded with the source of the \ilinva tool.

\subsection{Results}

We ran \ilinva on each example, first without disjunctive invariants (\ie, taking $n = 1$ in Procedure \Abduce) then 
with disjunctions of size $2$. The results are reported in Figure \ref{fig:results}.
For each example, we report whether our tool was able generate invariants allowing \whyt to check the soundness 
of all program assertions before the timeout, in which case we also report the time \ilinva took to do so (columns \texttt{T(C)} when generating conjunctions only and \texttt{T(D)} when generating implicants containing disjunctions).
We also report the number of candidate invariants that have been tried (columns \texttt{C(D)} and \texttt{C(D)}) and the number of abducible literals that were sent to the \gpid algorithm (column \texttt{Abd}). Note that the number of candidate invariants does not correspond to the number of SMT calls that are made by the system: those made by \gpid to generate these candidates are not taken into account.
The timeout is set to $20$ min.
For some of the examples that we deemed interesting, we allowed the algorithm to run longer. We report those cases by putting the results between parentheses.
Light gray cells indicate that the program terminates before the timeout without returning any solution, and dark gray cells indicate that the timeout was reached. Empty cells mean that the tool could not generate any candidate invariant.
The last column of both tables report the time \whyt takes to prove all the assertions of an example when correct invariants are provided.

The tests were performed on a computer powered by a dual-core Intel i5 processor running at 1.3GHz with 4 GB of RAM, under macOS 10.14.3.
We used \whyt version 1.2.0 and the SMT solvers \altergo (version 2.2.0), \cvcf (prerelease of version 1.7) and \zthree (version 4.7.1).

An essential point concerns the handling of local solver timeouts.
Indeed, most calls to the SMT solver in the abductive reasoning procedure will involve satisfiable formulas, and 
the  solvers usually take a lot of time to terminate on such formulas (or in the worst case will not terminate at all if the theory is not decidable, \eg, for problems involving first-order axioms).
We thus need to set a timeout after which a call will be considered as satisfiable (see Section \ref{sect:abd}).
Obviously, we neither want this timeout to be too high as it can significantly increase computation time, nor too low, since it could make us miss solutions.
We decided to set this timeout to $1$ second, independently of the solver used, after measuring the computation time of the \whyt verification conditions already satisfied (for which the solver returns \texttt{unsat}) across all benchmarks. We worked under the assumption that the computation time required to prove the other verification conditions when possible would be roughly similar. \label{sect:smtto}

\begin{figure}
\begin{minipage}[t]{0.5\textwidth}
\begin{center}
{\scriptsize
\begin{tabular}{|c|c|c|c|c|c|c|}
    \hline
    \rtnoth{} & \texttt{Abd} & \texttt{T(C)} & \texttt{C(C)} & \texttt{T(D)} & \texttt{C(D)} & \texttt{\whyt}\\
    \hline
    \texttt{O01} & \rtasol{36} & \rtasol{9.68} & \rtasol{7} & \rtasol{11.89} & \rtasol{10} & \rtasol{0.26}\\
    \hline
    \texttt{O02} & \rtasol{536} & \rtnosol{3'18.9} & \rtnosol{66} & \rtto{} & \rtto{1126} & \rtasol{0.45}\\
    \hline
    \texttt{O04} & \rtasol{108} & \rtnosol{50.47} & \rtnosol{32} & \rtasol{2'31.4} & \rtasol{156} & \rtasol{0.26}\\
    \hline
    \texttt{O05} & \rtasol{266} & \rtasol{1'9.07} & \rtasol{5} & \rtasol{1'3.2} & \rtasol{5} & \rtasol{0.31}\\
    \hline
    \texttt{O06} & \rtasol{390} & \rtnosol{6'13.6} & \rtnosol{56} & \rtnosol{18'5.1} & \rtnosol{552} & \rtasol{0.72}\\
    \hline
    \texttt{O07} & \rtasol{594} & \rtnosol{1'50.1} & \rtnosol{13} & \rtnosol{15'40.6} & \rtnosol{355} & \rtasol{0.38}\\
    \hline
    \texttt{O08} & \rtasol{210} & \rtasol{2'35.5} & \rtasol{61} & \rtasol{9'35.8} & \rtasol{528} & \rtasol{0.42}\\
    \hline
    \texttt{O09} & \rtasol{390} & \rtnosol{} & \rtnosol{0} & \rtnosol{} & \rtnosol{0} & \rtasol{0.56}\\
    \hline
    \texttt{O10} & \rtasol{90} & \rtasol{1'39.54} & \rtasol{65} & \rtnosol{12'56.9} & \rtnosol{0} & \rtasol{0.35}\\
    \hline
    \texttt{O11} & \rtasol{180} & \rtnosol{2'17.7} & \rtnosol{63} & \rtto{} & \rtto{942} & \rtasol{0.26}\\
    \hline
    \texttt{O12} & \rtasol{782} & \rtnosol{} & \rtnosol{0} & \rtnosol{} & \rtnosol{0} & \rtasol{0.53}\\
    \hline
    \texttt{O13} & \rtasol{296} & \rtnosol{2'4.5} & \rtnosol{0} & \rtto{} & \rtto{1621} & \rtasol{0.30}\\
    \hline
    \texttt{O14} & \rtasol{270} & \rtnosol{} & \rtnosol{0} & \rtnosol{} & \rtnosol{0} & \rtasol{0.34}\\
    \hline
    \texttt{O15} & \rtasol{36} & \rtnosol{32.53} & \rtnosol{21} & \rtto{} & \rtto{888} & \rtasol{0.27}\\
    \hline
    \texttt{O16} & \rtasol{60} & \rtnosol{12.54} & \rtnosol{8} & \rtasol{29.72} & \rtasol{32} & \rtasol{0.26}\\
    \hline
    \texttt{O17} & \rtasol{36} & \rtnosol{40.88} & \rtnosol{26} & \rtnosol{2'42.5} & \rtnosol{134} & \rtasol{0.33}\\
    \hline
    \texttt{O18} & \rtasol{38} & \rtnosol{58.49} & \rtnosol{38} & \rtnosol{6'53.3} & \rtnosol{0} & \rtasol{0.30}\\
    \hline
    \texttt{O19} & \rtasol{60} & \rtnosol{1'59.5} & \rtnosol{111} & \rtto{} & \rtto{1620} & \rtasol{0.31}\\
    \hline
    \texttt{O20} & \rtasol{546} & \rtto{} & \rtto{380} & \rtto{} & \rtto{870} & \rtasol{0.49}\\
    \hline
    \texttt{O21} & \rtasol{90} & \rtnosol{0.76} & \rtnosol{0} & \rtnosol{0.76} & \rtnosol{0} & \rtasol{0.38}\\
    \hline
    \texttt{O22} & \rtasol{270} & \rtasol{2'10.1} & \rtasol{20} & \rtasol{2'11.9} & \rtasol{20} & \rtasol{0.48}\\
    \hline
    \texttt{O23} & \rtasol{36} & \rtasol{4.6} & \rtasol{5} & \rtasol{4.7} & \rtasol{5} & \rtasol{0.28}\\
    \hline
    \texttt{O25} & \rtasol{60} & \rtnosol{1'23.4} & \rtnosol{20} & \rtnosol{2'38.4} & \rtnosol{44} & \rtasol{0.39}\\
    \hline
    \texttt{O26} & \rtasol{396} & \rtnosol{6'23.2} & \rtnosol{21} & \rtnosol{7'13.9} & \rtnosol{66} & \rtasol{0.83}\\
    \hline
    \texttt{O28} & \rtnoth{} & \rtnosol{2'3.9} & \rtnosol{137} & \rtasol{16'22.8} & \rtasol{1331} & \rtasol{0.31}\\
    \hline
    \texttt{O29} & \rtasol{61776} & \rtnosol{} & \rtnosol{0} & \rtnosol{} & \rtnosol{0} & \rtasol{0.65}\\
    \hline
    \texttt{O30} & \rtasol{36} & \rtasol{31.43} & \rtasol{26} & \rtasol{41.66} & \rtasol{45} & \rtasol{0.26}\\
    \hline
    \texttt{O31} & \rtasol{67050} & \rtnosol{} & \rtnosol{0} & \rtnosol{} & \rtnosol{0} & \rtasol{0.49}\\
    \hline
    \texttt{O32} & \rtasol{40} & \rtnosol{0.865} & \rtnosol{0} & \rtnosol{0.833} & \rtnosol{0} & \rtasol{0.50}\\
    \hline
    \texttt{O33} & \rtasol{90} & \rtnosol{1'11.3} & \rtnosol{12} & \rtnosol{1'19.9} & \rtnosol{21} & \rtasol{0.45}\\
    \hline
    \texttt{O34} & \rtasol{6768} & \rtnosol{0.798} & \rtnosol{0} & \rtnosol{0.79} & \rtnosol{0} & \rtasol{0.44}\\
    \hline
    \texttt{O35} & \rtasol{18} & \rtnosol{18.42} & \rtnosol{25} & \rtasol{2'7.9} & \rtasol{200} & \rtasol{0.26}\\
    \hline
    \texttt{O36} & \rtasol{61778} & \rtnosol{} & \rtnosol{0} & \rtnosol{} & \rtnosol{0} & \rtasol{1.09}\\
    \hline
    \texttt{O37} & \rtasol{36} & \rtnosol{0.752} & \rtnosol{0} & \rtnosol{0.769} & \rtnosol{0} & \rtasol{0.34}\\
    \hline
    \texttt{O38} & \rtasol{630} & \rtto{} & \rtto{444} & \rtnosol{3'54.4} & \rtnosol{0} & \rtasol{0.48}\\
    \hline
    \texttt{O39} & \rtasol{546} & \rtto{} & \rtto{1581} & \rtto{} & \rtto{1840} & \rtasol{0.40}\\
    \hline
    \texttt{O40} & \rtasol{272} & \rtnosol{} & \rtnosol{0} & \rtnosol{} & \rtnosol{0} & \rtasol{0.84}\\
    \hline
    \texttt{O41} & \rtnosol{} & \rtnosol{} & \rtnosol{0} & \rtnosol{} & \rtnosol{0} & \rtasol{0.37}\\
    \hline
    \texttt{O42} & \rtasol{271} & \rtnosol{1'50.4} & \rtnosol{25} & \rtto{} & \rtto{605} & \rtasol{1.12}\\
    \hline
    \texttt{O43} & \rtasol{60} & \rtasol{4.27} & \rtasol{2} & \rtasol{3.67} & \rtasol{2} & \rtasol{0.29}\\
    \hline
    \texttt{O44} & \rtnoth{} & \rtnosol{22.481} & \rtnosol{13} & \rtasol{5'7.8} & \rtasol{290} & \rtasol{0.35}\\
    \hline
    \texttt{O45} & \rtnoth{} & \rtnosol{} & \rtnosol{0} & \rtnosol{} & \rtnosol{0} & \rtasol{1.50}\\
    \hline
    \texttt{O46} & \rtnoth{} & \rtto{} & \rtto{513} & \rtto{} & \rtto{813} & \rtasol{0.61}\\
    \hline
\end{tabular}
}
\end{center}
\end{minipage}
\begin{minipage}[t]{0.5\textwidth}
\begin{center}
{\scriptsize
\begin{tabular}{|c|c|c|c|c|c|c|}
    \hline
    \rtnoth{} & \texttt{Abd} & \texttt{T(C)} & \texttt{C(C)} & \texttt{T(D)} & \texttt{C(D)} & \texttt{\whyt}\\
    \hline
    \texttt{509} & \rtasol{130} & \rtasol{(1h50')} & \rtasol{(95)} & \rtto{} & \rtto{0} & \rtasol{0.66}\\
    \hline
    \texttt{534} & \rtasol{172k} & \rtto{} & \rtto{8} & \rtto{} & \rtto{0} & \rtasol{0.62}\\
    \hline
    \texttt{H04} & \rtasol{120} & \rtasol{2'54.8} & \rtasol{223} & \rtto{} & \rtto{1383} & \rtasol{0.31}\\
    \hline
    \texttt{H05} & \rtasol{1260} & \rtnosol{} & \rtnosol{0} & \rtnosol{} & \rtnosol{0} & \rtasol{0.37}\\
    \hline
    \texttt{list0} & \rtasol{60} & \rtasol{6'30.4} & \rtasol{77} & \rtto{} & \rtto{1722} & \rtasol{0.40}\\
    \hline
    \texttt{list1} & \rtasol{20} & \rtasol{40.82} & \rtasol{3} & \rtasol{3'26.2} & \rtasol{385} & \rtasol{0.47}\\
    \hline
    \texttt{list2} & \rtasol{720} & \rtto{} & \rtto{40} & \rtto{} & \rtto{0} & \rtasol{0.40}\\
    \hline
    \texttt{list3} & \rtasol{126} & \rtnosol{3'35.1} & \rtnosol{11} & \rtto{} & \rtto{930} & \rtasol{0.44}\\
    \hline
    \texttt{list4} & \rtasol{816} & \rtto{} & \rtto{18} & \rtto{} & \rtto{0} & \rtnosol{}\\
    \hline
    \texttt{list5} & \rtasol{468} & \rtto{} & \rtto{22} & \rtto{} & \rtto{0} & \rtasol{0.44}\\
    \hline
    \texttt{array0} & \rtnosol{} & \rtnosol{} & \rtnosol{0} & \rtnosol{} & \rtnosol{0} & \rtasol{0.72}\\
    \hline
    \texttt{array1} & \rtnosol{} & \rtnosol{} & \rtnosol{0} & \rtnosol{} & \rtnosol{0} & \rtasol{0.50}\\
    \hline
    \texttt{array2} & \rtnosol{} & \rtnosol{} & \rtnosol{0} & \rtnosol{} & \rtnosol{0} & \rtasol{0.50}\\
    \hline
    \texttt{array3} & \rtnosol{} & \rtnosol{} & \rtnosol{0} & \rtnosol{} & \rtnosol{0} & \rtasol{0.82}\\
    \hline
    \texttt{expo0} & \rtasol{171} & \rtasol{(6h36')} & \rtasol{(9)} & \rtto{} & \rtto{0} & \rtasol{0.40}\\
    \hline
    \texttt{expo1} & \rtasol{2130} & \rtto{} & \rtto{0} & \rtto{} & \rtto{0} & \rtnosol{}\\
    \hline
    \texttt{square} & \rtasol{705} & \rtto{} & \rtto{62} & \rtto{} & \rtto{148} & \rtnosol{}\\
    \hline
    \texttt{real0} & \rtasol{965} & \rtto{} & \rtto{81} & \rtto{} & \rtto{213} & \rtasol{0.55}\\
    \hline
    \texttt{real1} & \rtasol{965} & \rtto{} & \rtto{73} & \rtto{} & \rtto{115} & \rtasol{0.55}\\
    \hline
    \texttt{real2} & \rtasol{240} & \rtto{} & \rtto{9} & \rtto{} & \rtto{2} & \rtasol{0.40}\\
    \hline
    \texttt{realO0} & \rtasol{36} & \rtasol{4'9.6} & \rtasol{25} & \rtasol{5'32.18} & \rtasol{40} & \rtasol{0.47}\\
    \hline
    \texttt{realS} & \rtasol{66} & \rtasol{1'5.3} & \rtasol{5} & \rtasol{1'0.1} & \rtasol{5} & \rtasol{0.33}\\
    \hline
    \texttt{real3} & \rtasol{17460} & \rtto{} & \rtto{0} & \rtto{} & \rtto{0} & \rtnosol{}\\
    \hline
    \texttt{BM} & \rtasol{1260} & \rtasol{3'15.2} & \rtasol{74} & \rtto{} & \rtto{33} & \rtasol{3.35}\\
    \hline
    \texttt{Scmp} & \rtnosol{} & \rtnosol{} & \rtnosol{0} & \rtnosol{} & \rtnosol{0} & \rtasol{0.83}\\
    \hline
    \texttt{Dmd} & \rtasol{42} & \rtto{} & \rtto{6} & \rtto{} & \rtto{0} & \rtasol{1'44.9}\\
    \hline
    \texttt{B00} & \rtasol{639k} & \rtto{} & \rtto{0} & \rtto{} & \rtto{0} & \rtasol{0.76}\\
    \hline
    \texttt{DIV0} & \rtasol{560} & \rtnosol{3'58} & \rtnosol{35} & \rtto{} & \rtto{534} & \rtasol{0.83}\\
    \hline
    \texttt{DIV1} & \rtasol{310} & \rtasol{14.6} & \rtasol{19} & \rtasol{14.6} & \rtasol{19} & \rtasol{0.44}\\
    \hline
    \texttt{DIVE} & \rtasol{42250} & \rtto{} & \rtto{0} & \rtto{} & \rtto{0} & \rtnosol{}\\
    \hline
\end{tabular}
}
\end{center}
\end{minipage}
\caption{Experimental Results\label{fig:results}}
\end{figure}

\subsection{Discussion}

\label{sect:discuss}

As  can be observed, \Ilinva is able to generate solutions for a wide range of theories, although the execution 
time is usually high.
The number of invariant candidates is relatively  high, which has a major impact on the efficiency and scalability 
of the approach.

When applied to examples involving arithmetic invariants, the program is rather slow, compared to 
the approach based on minimization and quantifier elimination \cite{DBLP:conf/oopsla/DilligDLM13}. 
This is not surprising, since 
it is very unlikely that  a purely generic 
approach based on a model-based tree exploration algorithm involving many calls to an SMT solver can possibly compete with a more specific procedure exploiting particular properties of the 
considered theory. We also wish to emphasize that the fact that our framework is based on an external program verification system (which itself calls external solvers) involves a significant overcost compared to a more integrated approach: for instance, for the {\tt Oxy} examples (taken from \cite{DBLP:conf/oopsla/DilligDLM13}), the time used by \whyt to check the verification conditions once the correct invariants have been generated is often greater than the total time reported in \cite{DBLP:conf/oopsla/DilligDLM13} for computing the invariants 
and checking all properties. Of course, our choice also has clear advantages in terms of genericity, generality and  evolvability.

When applied to theories that are harder for SMT solvers, the algorithm can still generate satisfying invariants.
However, due to the high number of candidates it tries, combined with the heavy time cost of a verification (which can be several seconds), it may take some time to do so.

The number of abducible literals has a strong impact on the efficiency of the process, leading to timeouts
when the considered program contains many variables or function/predicate symbols.
 It can be observed that the abduction depth is rather low in all examples ($1$ or $2$).

Our prototype has some technical limitations that have a significant impact on the time cost of the execution.
For instance, we use  \smtlib files for communication between \gpid and \cvcf or \zthree, instead of using the available APIs. 
We  went back to this solution, which is clearly not optimal for performance, because we experienced  many problems
coping with the numerous changes in the specifications when updating the solvers to always use the latest versions. 
The fact that \whyt is taken as a black box also yields some time consumption, first in 
the (backward and forward) translations (\eg, to associate program variables to their logical counterparts), but also in the verification tasks, which have to be rechecked from the beginning each time an invariant is updated.
Our aim in the present paper was not to devise an efficient system, but rather to assess the feasability and usefulness of this approach.
Still, the cost of the numerous calls to the SMT solvers and the size of the search tree of the abduction procedure remain the bottleneck of the approach, especially for hard theories (\eg, non-linear arithmetics) for which most calls with satisfiable formulas yield to a local timeout (see Section \ref{sect:smtto}).

\section{Conclusion and Future Work}

By combining our  generic system \gpid for abductive reasoning modulo theories with the \why platform to generate verification conditions, we obtained a tool to check properties of \whyml  programs, which is able to 
compute loop invariants in a purely automated way.

The main drawback of our approach is that the set of possible abducible literals is large, yielding 
a huge search space, especially if disjunctions of {\implicant}s are considered. 
Therefore, we believe that our system  in its current state 
is mainly useful when used in an interactive way. For instance, the user could 
provide the properties of interest for some of the loops and let the system automatically 
compute suitable invariants by combining these properties, or the program could rely on the user to choose between 
different solutions to the abduction problem before applying the strengthening.
Beside, it is also useful 
for dealing with theories for which no specific abductive reasoning procedure exists, especially
for reasoning in the presence of user-defined symbols or axioms.

In the future, we will focus on the definition of suitable heuristics for  automatically  selecting
abducible literals and ordering them, to reduce the search space and 
avoid backtracking. The number of occurrences of symbols should be taken into account, 
as well as properties propagating from previous invariant strengthening.
A promising approach is to use dynamic program analysis tools to select relevant abducibles.
It would also be interesting to adapt the \gpid algorithm to explore the search space width-first, to ensure that simplest solutions are always generated first. Another option is to give \Ilinva a more precise control on the \gpid algorithm, \eg, to explore some branches more deeply, based on  information related to the verification conditions. 
\gpid could also be tuned to generate disjunctions of solutions in a more efficient way. 

From a more technical point of view, a tighter integration with the \why platform would certainly be beneficial, as explained in Section \ref{sect:discuss}.
The framework could be extended to handle procedures and functions (with pre- and -post conditions).

 A current limitation of our tool is that it cannot handle problems in which \whyt relies on a combination of different solvers to check the desired properties. In this case, \Ilinva cannot generate the invariants, 
as the same SMT solver is used for each abduction problem (trying all solvers in parallel on every problem would be  possible in theory but this would greatly 
increase the search space).  This problem could be overcome by using heuristic approaches 
to select the most suited solver for a given problem.

From a theoretical point of view, it would be interesting to investigate the completeness of our approach.
It is clear that no general completeness result  possibly holds, due to usual theoretical limitations, 
however, if we assume that a program $\aprog' \sameprog \aprog$ such that 
$\vcgen{\aprog'}$ is valid exists, does the call $\Ilinva(\aprog)$ always succeed?
This of course would require that the invariants in $\aprog'$ can be constructed from abducibles occurring in the set returned by the procedure $\extractabd$.

\bibliography{frocos}
\bibliographystyle{abbrv}

\end{document}